\documentclass[12pt]{elsarticle}

\usepackage{amsmath,amsfonts,amssymb,amsthm,bbm,color}

\def\uno{\mathbbm 1}
\def\vol{\mathrm{vol}}
 
\def\bar{\overline}
\def\fro{\mathrm{F}}

\newtheorem{theorem}{Theorem}[section]
\newtheorem{lemma}[theorem]{Lemma}
\newtheorem{corollary}[theorem]{Corollary}
\newtheorem{definition}[theorem]{Definition}
\newtheorem{remark}[theorem]{Remark}
\newtheorem{example}[theorem]{Example}

\makeatletter
\def\ps@pprintTitle{%
 \let\@oddhead\@empty
 \let\@evenhead\@empty
 \def\@oddfoot{}%
 \let\@evenfoot\@oddfoot}
\makeatother

\begin{document}

\title{A modularity based spectral method 
for simultaneous community and anti-community detection}

\author[dima]{Dario Fasino\corref{cor1}}
\ead{dario.fasino@uniud.it}

\address[dima]{Department of Mathematics, Computer Science and Physics,
University of Udine, Udine, Italy.}

\author[sb]{Francesco Tudisco}
\ead{francesco.tudisco@math.unipd.it}

\address[sb]{Department of Mathematics, University of Padua, Padua, Italy.}
\cortext[cor1]{Corresponding author}


\begin{abstract}
In a graph or complex network, communities and anti-communities 
are node sets whose modularity attains extremely large values,
positive and negative, respectively.  We consider the simultaneous
detection of communities and anti-communities, 
by looking at spectral methods based on various 
matrix-based definitions of the modularity of a vertex set. 
Invariant subspaces associated to extreme eigenvalues
of these matrices provide indications
on the presence of 
both kinds of modular structure in the network.
The localization of the relevant invariant subspaces can be estimated by looking at particular
matrix angles based on 
Frobenius inner products. 
\end{abstract}

\maketitle

\section{Introduction}

This paper addresses 
the problem of 
grouping nodes of a network into communities and anti-communities, 
possibly emerging from a neutral background. 
A community is roughly defined as a set of nodes being highly connected inside and poorly connected with the rest of the graph. Conversely, an anti-community is a node set being loosely connected inside but having many external connections.  Revealing these structures in data and networks is a challenging and relevant problem which has applications in many disciplines, ranging from computer science to physics and several natural and social sciences \cite{Anti-modularity,Estrada-negative,holme2003network,mercado2016clustering,HighamLockAndKey}.

In order to address this problem from the mathematical point of view one needs a quantitative definition of what a community and an anti-community is. To this end several merit functions have been introduced in the recent literature. 
A very popular and fruitful idea is based on the concept of modularity, originally introduced for community detection 
in the statistical mechanics literature
\cite{Newman-eigenvector,newman2004finding}. 
The modularity measure of a set of nodes $S\subset V$ in a graph $G=(V,E)$ quantifies the difference between the actual weights of edges in $S$ with respect to the expected weight, if edges were placed at random according to a prescribed \textit{null model}. 
The modularity-based criterion for community detection thus identifies a  subset $S$ as a community if its modularity measure is ``large'',  
and as an anti-community if its modularity measure is ``small''.
The (anti-)community detection problem thus boils down to a combinatorial optimization problem, whose solution is typically approximated through a matrix-based technique which exploits the spectrum of a suitably defined \textit{modularity matrix}.

In this work we show that dominant eigenvalues of generalized modularity matrices can be used to simultaneously look for communities and anti-communities. We propose a spectral method based on the eigenspaces associated with those eigenvalues and 
relate its performance 
to certain matrix angles. 
We then analyze the stochastic block model,
one of the most useful generative models in community detection, to obtain indications on the average performance of 
the proposed method.
To that goal, we characterize 
the dominant eigenvalues and eigenvectors of the average modularity matrix in the model.
A couple of numerical experiments are included to validate the proposed 
computational strategy.

\subsection{Notations and preliminaries} 

In the sequel
we give a brief review of standard concepts and symbols from algebraic graph theory that we will use throughout the paper.
We assume that $G = (V, E)$ 
is a 
finite, undirected graph
where $V$ and $E$ are the vertex and edge 
sets, respectively. We will
identify $V$ with $\{1,\dots,n\}$. 
We denote adjacency of vertices $i$ and $j$ 
as $ij\in E$. We allow positive weights on both the vertex and the edge sets which we denote by $\mu:V\to\mathbb R_+$, $i\mapsto \mu(i)$ and $w:E\to\mathbb R_+$, $ij\mapsto w_{ij}$, respectively.

The symbol $A$ denotes the adjacency matrix of $G$, 
that is, $A = (a_{ij})$
where $a_{ij} = w_{ij}$ if $ij \in E$, and $a_{ij}=0$ otherwise.
In particular, $A$ is a symmetric, componentwise nonnegative matrix.
The (generalized) degree of vertex $i\in V$ is
$d_i = \sum_{j=1}^n w_{ij}$, and
$\uno$ denotes an all-one vector
whose dimension depends on the context.
With this notation, the degree vector is $d = A\uno$. 

The subgraph induced by a set $C\subseteq V$ is the 
graph $G(C)$ whose adjacency matrix is the submatrix of $A$
whose row and column indices are in $C$.
The cardinality of $C$ is denoted by $|C|$, 
and its weight by $\mu(C)=\sum_{i\in C}\mu(i)$. For consistency with other works by various authors, a special notation is reserved for the case $\mu(i)=d_i$ where we write $\vol(C) = \sum_{i \in C}d_i$ for the \textit{volume} of $C$. Correspondingly,  
$\vol(V)= \sum_{i \in V}d_i$ denotes the volume of the whole graph. 
Moreover, we denote by
$\bar C$ the complement $V \setminus C$,
and let $\uno_C$ be its \textit{characteristic vector}, defined as $(\uno_C)_i =1$ if $i \in C$ and $(\uno_C)_i=0$ otherwise,
so that $\vol(C) = \uno_C^Td$. 

The Frobenius inner product of two matrices $A,B\in\mathbb{R}^{n\times n}$ is $\langle A,B\rangle = \mathrm{trace}(A^TB)$,
and the associated matrix norm is $\|A\|_\fro = \langle A,A\rangle^{1/2}$.   
For later reference, we recall the Hoffman--Wielandt
theorem for singular values and the Eckart--Young theorem,
see e.g., \cite{GVL, MR1061154}.

\begin{theorem}[Hoffman--Wielandt]   \label{thm:hoffman-wielandt-sv}
Let $A$ and $B$ be two $n\times n$ matrices. Denote by 
$s_i(A)$ and $s_i(B)$ their singular values
arranged in nonincreasing order. Then
$$
   \sum_{i=1}^n (s_i(A)-s_i(B))^2 
   \leq \|A-B\|_\fro^2 .
$$
\end{theorem}

\begin{theorem}[Eckart--Young]   \label{thm:eckart-young}
Let $A = U\Lambda U^T$ be a symmetric matrix
with eigenvalues arranged in nonincreasing modulus,
$|\lambda_1|\geq |\lambda_2| \geq \ldots \geq |\lambda_n|$.
For any integer $1\leq k \leq n$ let $A_k = U_k\Lambda_kU_k^T$
where $U_k$ is the $n\times k$ matrix 
formed by the first $k$ columns of $U$ and
$\Lambda_k = \mathrm{Diag}(\lambda_1,\ldots,\lambda_k)$.
Then
$$
   \| A - A_k \|_\fro^2 = \min_{\mathrm{rank}(B) \leq k}\|A - B\|_\fro^2 
   = \textstyle{\sum_{i=k+1}^n \lambda_i^2} .
$$
\end{theorem}

\subsection{Matrix projections and angles}

Let $\mathcal{S}$ be a linear subspace of $\mathbb{R}^{n\times n}$.
The orthogonal projection of a matrix $A$ onto $\mathcal{S}$
with respect to the Frobenius inner product is
$$ 
   P = \mathcal{P}(A,\mathcal{S}) \quad \Longleftrightarrow
   \quad \| A - P \|_\fro = \min_{B\in\mathcal{S}} \| A - B\|_\fro .
$$ 
Furthermore,
we define the angle between $A$ and $\mathcal{S}$
through the trigonometric functions
$$
   \sin(A,\mathcal{S}) = \frac{\|A - P\|_\fro}{\|A\|_\fro} ,
   \qquad
   \cos(A,\mathcal{S}) = \frac{\|P\|_\fro}{\|A\|_\fro} ,
$$
where $P = \mathcal{P}(A,\mathcal{S})$.

For any fixed integer $1\leq k \leq n-1$ let 
$\mathbb{O}^{n\times k}$ be the set of all $n\times k$ matrices
with orthonormal columns. For any    
$X = [x_1,\ldots,x_k]\in \mathbb{O}^{n\times k}$
we denote by $\mathcal{H}(X)$ the matrix subspace
$$
   \mathcal{H}(X) = \left\{ H = \sum_{i=1}^k \tau_i x_ix_i^T ,
   \tau_i\in\mathbb{R} \right\} .
$$
Equivalently, any matrix in $\mathcal{H}(X)$
admits the factorization $XT X^T$ where
$T = \mathrm{Diag}(\tau_1,\ldots,\tau_k)$.
Simple arguments produce the explicit expression
$$
    \mathcal{P}(A,\mathcal{H}(X)) = \sum_{i=1}^k \tau_i x_ix_i^T ,
    \qquad \tau_i = x_i^T A x_i .
$$
Moreover, we denote by $\mathcal{K}(X)$ the matrix subspace
\begin{equation}   \label{eq:defK}
   \mathcal{K}(X) = \left\{ K = XSX^T , \
   S = S^T \in \mathbb{R}^{k\times k} \right\} .
\end{equation}

Equivalently, $K \in\mathcal{K}(X)$ 
if and only if $X^T KX$ is a symmetric $k \times k$ matrix.
Note that $\mathcal{H}(X)\subset\mathcal{K}(X)$.
Simple arguments produce the explicit expression
$$
    \mathcal{P}(A,\mathcal{K}(X)) = XX^T A XX^T .
$$

The following result is an easy consequence of the Eckart--Young theorem; we omit the simple proof, which is based on the fact that,
for arbitrary $X\in\mathbb{O}^{n\times k}$, both
$\mathcal{H}(X)$ and $\mathcal{K}(X)$ describe the set 
of all $n\times n$ 
symmetric 
matrices with rank not larger than $k$.

\begin{theorem}  \label{thm:XYZ}
Let $A = U\Lambda U^T$ be a symmetric matrix
with eigenvalues arranged in nonincreasing modulus,
$|\lambda_1|\geq |\lambda_2| \geq \ldots \geq |\lambda_n|$.
For any fixed integer $1\leq k \leq n-1$, suppose that $|\lambda_k|> |\lambda_{k+1}|$. Let $X_H,X_K\in\mathbb{O}^{n\times k}$ be solutions of the variational problems
$$
   \max_{X\in\mathbb{O}^{n\times k}} \cos(A,\mathcal{H}(X)),
   \qquad
   \max_{X\in\mathbb{O}^{n\times k}} \cos(A,\mathcal{K}(X)),   
$$
respectively. Then there exist orthogonal matrices
$Z_1,Z_2\in\mathbb{R}^{k\times k}$ such that
$$
   U_k = X_HZ_1 = X_KZ_2 ,
$$
where $U_k\in\mathbb{O}^{n\times k}$ is formed by the first $k$ columns of $U$. 
In both cases, the projection of $A$
onto the respective spaces is $A_k = U_k\Lambda_kU_k^T$ 
where $\Lambda_k = \mathrm{Diag}(\lambda_1,\ldots,\lambda_k)$.
\end{theorem}

\section{Communities and anti-communities}

Discovering the presence of communities in a network is a central problem in data analysis and modern network science. Although 
there is no clear-cut definition of what a community is,
common sense suggests that a set of nodes can be recognized
as a community inside a given network if 
those nodes are tightly connected internally, and loosely connected with the surrounding ones. A rather successful formalization
of this informal argument is originally due to Newman and Girvan
\cite{newman2004finding},
who proposed a specific measure, called modularity,
to quantify the ``strength'' of a set of nodes as a
community. 
Such measure is based on the argument that $C\subseteq V$ 
can be recognized as a community if the subgraph 
induced by $C$ contains more edges than expected, if edges were placed at random according to a certain random graph model.
A very noticeable example is the 
Newman--Girvan modularity, which is defined for any
$C\subseteq V$ as
\begin{equation}   \label{eq:NG}
   Q(C) = \uno_C^T\widehat M\uno_C , \qquad
   \widehat M = A - dd^T\!/\vol(V) .
\end{equation}
Here, $\uno_C^TA\uno_C$ quantifies the overall weight of
edges internal to the subgraph $G(C)$ while 
$(\uno_C^Td)^2/\vol(V)$ is an {\em a priori} 
estimate of the former quantity
according to the so-called Chung--Lu random graph model
\cite{chung2006complex}.  
The 
Newman--Girvan modularity enjoys nice properties, for example
$Q(C) = Q(\bar C)$, and $Q(V) = 0$.
We point the reader to \cite{fasino2014algebraic}
for an extensive analysis of spectral properties of that matrix $\widehat M$.

A number of different variants of the modularity measure have been considered afterwards, see e.g., \cite{bolla2011penalized, fasino2016generalized, multires1, tudisco2016community}. 
All of them can be expressed as 
$\uno_C^T\widehat M\uno_C$ by means of a suitably defined
(generalized) modularity matrix $\widehat M$.
Within 
this setting, the detection of $k$ largest communities in a prescribed network, or the  
partitioning of its vertex 
set into $k$ pairwise 
disjoint communities translates naturally into the maximization of $\sum_{i=1}^k Q(C_i)$ under appropriate conditions on the family $\mathcal{C} = \{C_1,\ldots,C_k\}$.
A major alternative formulation is related to normalized versions of the modularity,
\begin{equation}   \label{eq:q(C)}   
   q(C) = Q(C)/\mu(C)
\end{equation}
where $Q(C)$ is as before and $\mu(C)$ 
is an additive measure of the set $C$, which is used as a balancing function to promote small node groups. The two most common measures are
$\mu(C) = |C|$ and $\mu(C) 
= \vol(C)$, see e.g., \cite{bolla2011penalized}.
The use of 
the normalized modularity $q(C)$ 
has several advantages 
because of its more reliable underlying algebraic structure
and its ability to localize small communities, despite the known tendency of the standard (unnormalized) modularity measures to overlook small groups 
\cite{fortunato2007resolution, tudisco2016community}. 
We also mention that 
the normalization term 
$\mu(C) = \tau|C| + \vol(C)$
has been introduced in a closely related context to cope with
sparse networks with strong degree heterogeneity
\cite{ChungCOLT2012,QinNIPS2013}; 
there, $\tau > 0$ is a constant to be tuned 
in order to minimize the variance of certain statistical estimators.

Besides communities, another subgraph-level structure 
of interest in network analysis is that of
anti-communities. Probably, the first occurrence of the term
``anti-community'' can be traced back to \cite{Newman-eigenvector}
with reference to vertex sets having $Q(C)\ll 0$.
Roughly speaking, a set of nodes is said to have 
an anti-community structure 
if internal connections are fewer than
those expected by chance.
According to \cite{Anti-modularity},
in an anti-community nodes have most of their connections outside their group and have no or fewer connections with the members within the same group.
For example, an almost bipartite graph
(a bipartite graph possibly containing a few erroneous edges) 
is made of two anti-communities.
These informal definitions can be suitably formalized by 
means of modularity concepts. 
Actually, the detection of a bipartite structure in a given graph
with vertex set $V$ can be handled by the minimization of the Newman--Girvan
modularity over subsets of $V$; owing to the equality
$Q(C)=Q(\bar C)$, if $C$ is a set of minimum modularity then
$\{C,\bar C\}$ gives the bipartite structure of the graph.
In fact, if the subgraph induced by $C$ contains no edges, then 
$\vol(C) = \vol(\bar C) = \vol(V)/2$ and
$$
   Q(C) = -\vol(C)^2/\vol(V) = -\vol(V)/4 ,
$$
which is the lower bound for $Q$.
Further, literature
concerned with anti-community structures addresses
both the detection of almost-bipartitiveness,
see e.g., \cite{Newman-eigenvector,HighamBipartite},
and the presence of multiple, possibly overlapping anti-communi\-ties
\cite{bolla2011penalized,Anti-modularity,Estrada-negative,HighamLockAndKey}.

In the present work we address the problem of detecting simultaneously 
the presence of communities and anti-communities in a given network. 
To that goal, we extend the spectral approach
that has been adopted by various authors in community or
anti-community detection, as mentioned above.
Henceforth, we agree that a
node set $C\subset V$ is an anti-community if 
$q(C) < 0$ and ``large'' (compared to other subsets), 
and a {\em module} if $q(C)^2$ is ``large'',
where $q(C)$ is a normalized modularity function 
as in \eqref{eq:q(C)}.

\section{Modularity-based detection of communities and anti-communities}

Hereafter, we consider the following framework:
we are given a (generalized)
modularity matrix $\widehat M$ 
such that $Q(C) = \uno_C^T\widehat M\uno_C$
and the diagonal matrix 
$W = \mathrm{Diag}(\mu(1),\ldots,\mu(n))$. To 
any subset $C\subseteq V$ we associate
the \textit{measure vector} $\chi_C$ defined as
$$
   \chi_C = W^{1/2}\uno_C / \|W^{1/2}\uno_C\|_2 .
$$
Consequently we have $\chi_C^T \chi_C = 1$ and  
$$ 
   \chi_C^T M \chi_C = q(C) ,
   \qquad M = W^{-1/2}\widehat MW^{-1/2} .
$$

\begin{example}   \label{ex:NG}
Let $Q(C)$ be the Newman--Girvan modularity \eqref{eq:NG}.  
\begin{itemize}

\item
If $\mu(j) = 1$ for all $j$ then  $M = \widehat M$,
$\chi_C = \uno_{C}/|C|^{1/2}$ and
$$
   q(C) = \frac{\uno_{C}^T\widehat M\uno_{C}}{|C|}  
   = \frac{Q(C)}{|C|} ;
$$
\item
if $\mu(j) = d_j$ then $\chi_C = W^{1/2}\uno_{C}/\vol(C)^{1/2}$
and $q(C) = Q(C)/\vol(C)$.
\end{itemize}
A statistics based scrutiny of these two ``penalized'' versions
of $Q(C)$ has been carried out in \cite{bolla2011penalized}.
\end{example}

Consider a family $\mathcal{C} = \{C_1,\ldots, C_k\}$ of pairwise disjoint subsets
of $V$ with $\cup_{i=1}^k C_i \subseteq V$, and let $X = [\chi_1,\ldots,\chi_k]\in\mathbb{O}^{n\times k}$
where $\chi_i$ is the measure
vector of $C_i$. 
The problem of recognizing the $k$ 
largest modules in $G$
can be obviously stated as the 
maximization   
over all such $\mathcal{C}$ 
of the quantity
$$
   \sigma(\mathcal{C}) = \sum_{i=1}^k q(C_i)^2 . 
$$
Owing to the equivalent expressions
$$
   \sigma(\mathcal{C}) 
   = \sum_{i=1}^k (\chi_i^TM\chi_i)^2 
   = \| \mathcal{P}(M,\mathcal{H}(X))\|_\fro^2
   = \| M \|_\fro^2
   \cos(M,\mathcal{H}(X))^2 ,
$$
a continuous relaxation of that combinatorial problem consists 
of  
computing 
\begin{equation}   \label{eq:maxourq}
   X^* = \mathrm{arg}
   \max_{X\in \mathbb{O}^{n\times k}} \cos(M,\mathcal{H}(X)) .
\end{equation}
As recalled in Theorem \ref{thm:XYZ},
any solution of \eqref{eq:maxourq}
is related to the $k$ ``dominant'' eigenvectors of $M$.
In fact, given the spectral decomposition $M = U\Lambda U^T$ 
with $U=[u_1,\dots,u_n]$,
consider the eigenvalues of $M$ 
in nonincreasing moduli:
\begin{equation}   \label{eq:lambdaM}
   |\lambda_1 | \geq |\lambda_2| \geq \cdots \geq |\lambda_n| .
\end{equation}
If $|\lambda_k | > |\lambda_{k+1}|$ then
any solution $X^*$ of \eqref{eq:maxourq} is
an orthogonal transform of $U_k = [u_1,\ldots,u_k]$,  
and the orthogonal projection of $M$ onto 
$\mathcal{H}(X^*)$ is
$$
   H = \mathcal P(M,\mathcal H(X^*))
   =U_k \mathrm{Diag}(\lambda_1,\ldots,\lambda_k) 
   U_k^T .
$$
Note that $\|H \|_\fro^2 = \sum_{i=1}^k \lambda_i^2$ and
$\| M - H \|_\fro^2 = \sum_{i=k+1}^n \lambda_i^2$,
hence
\begin{equation}   \label{eq:maxcos}
   \cos(M,\mathcal{H}(X^*))^2 = 
   \textstyle{(\sum_{i=1}^k \lambda_i^2)/(\sum_{i=1}^n \lambda_i^2)} .
\end{equation}
Thus the presence of a ``good'' set of $k$ modules is 
indicated by the presence of $k$ ``large'' eigenvalues
in the spectrum of $M$.
Moreover, the  
subsets $C_1,\ldots,C_k$ can be recovered
from $U_k$ by virtue of the coincidence of
the projections $\mathcal{P}(M,\mathcal{H}(X^*))$ and
$\mathcal{P}(M,\mathcal{K}(X^*))$
as described 
here below.

Consider the matrix space $\mathcal{K}(X)$ defined in \eqref{eq:defK} with $X = [\chi_1,\ldots,\chi_k]$. Assuming for the moment that $\cup_{i=1}^k C_i = V$,
any matrix in that space has the
following block structure, up to a row and column renumbering:
$$
   K = K^T = 
   \begin{pmatrix} K_{11} & \cdots & K_{1k}\\
   \vdots & \vdots & \vdots \\
   K_{k1} & \cdots & K_{kk} \end{pmatrix} . 
$$
Here, block $K_{ij}$ has order $|C_i|\times |C_j|$ and rank one.
In fact, let $S = X^TKX = (s_{ij})$
and denote by $\hat \chi_i$ the nonzero part 
(i.e., the support) of $\chi_i$.
Then, $K_{ij} = s_{ij} \hat \chi_i\hat \chi_j^T$.
For example, if $\chi_i = \uno_{C_i}/\sqrt{|C_i|}$ then
$K_{ij} = (s_{ij}/\sqrt{|C_i||C_j|})\uno\uno^T$.
Moreover, given the spectral decomposition 
$S = Y^TDY$
where $Y = [y_1,\ldots,y_k]$ is orthogonal,
the matrix $K$ admits the spectral factorization 
$$
   K = X Y^T D Y X^T .
$$
The matrix of eigenvectors relative to nonzero eigenvalues of $K$ can be partitioned as follows:
$$
   XY^T = \begin{pmatrix}
   \hat \chi_1 y_1^T \vspace{2mm}\\ \hat \chi_2 y_2^T\vspace{2mm} \\ \vdots \vspace{2mm}\\ \hat \chi_k y_k^T 
   \end{pmatrix}
   \hspace{-2mm}
   \begin{array}{c}
   \big\} |C_1| \vspace{2mm}\\ \big\} {\tiny |C_2|} \vspace{2mm}\\ \vdots \vspace{2mm}\\ \big\} |C_k| 
   \end{array}
$$
that is, rows with indices in the same $C_i$ are parallel,
while rows belonging to different $C_i$'s are orthogonal.
In particular, if $i\in C_j$ then the norm of the $i$-th row
is $\sqrt{\mu(i)/\mu(C_j)}$, so that 
scaling that row by $1/\sqrt{\mu(i)}$ makes the norm
depend only on the block index $j$.
Thus the partitioning $\mathcal{C} = \{C_1,\ldots,C_k\}$ can be recovered
from $XY^T$ by clustering its rows according to their angles
and lengths.
On the other hand, if $\mathcal{C}$ does not cover $V$
then, up to a row renumbering, 
the block structure of the matrix $K$ above appears in the 
upper left corner of a larger matrix bordered by null rows and columns. Indeed, in this case
$X$ has a trailing block of 
null rows, which remains unchanged by right-multiplication
times orthogonal matrices.
Those null rows indicate nodes not belonging to any $C_i$'s, 
while the nonzero rows convey the same information as before.

Clearly, if $\cos(M, \mathcal H(X))=1$ then $M\in \mathcal H(X)\subseteq \mathcal K(X)$, 
$U_k=XY^T$
and the modules of $\mathcal{C}$
are inscribed in 
the rows of $U_k$.
In a more realistic setting,
by a continuity argument if $\cos(M,\mathcal{H}(X))$
does not exactly equal one but is ``close to one'' then $U_k$ can be regarded
as a perturbation of $XY^T$
in the sense that there is an orthogonal matrix $Z$
such that $U_k - XZ$ has small norm.
Precise statements can be obtained from Davis--Kahan 
``$\sin\theta$'' theorem \cite[Thm.\ 3.6]{MR1061154};
the ensuing Corollary \ref{cor:sintheta}
provides a result in that spirit.

Hence, a spectral method to locate
$k$ ``dominant'' modules of a given graph or network,
possibly not covering the whole vertex set, can be based on
the computation of $k$ extreme eigenpairs
of $M$, and the use of 
a clustering algorithm to group rows of $U_k$ according to 
their position in $\mathbb{R}^k$. 
For that task, a popular idea in the clustering literature 
is to adopt the $k$-means algorithm, see e.g., 
\cite{bolla2011penalized,QinNIPS2013,SpectralClustering}.
However, since that algorithm always returns a partitioning
which covers $V$, we
keep in view that the number of disjoint subsets
can be equal to the number of considered eigenvalues plus one; 
the supplemental set can be interpreted as a ``background''
in the graph from where the dominant modules emerge.
This possibility is substantiated in the analysis of 
the stochastic block model in section \ref{sec:SBM}
and illustrated by a couple of numerical examples in 
section \ref{sec:numerical}.

\section{Analysis of the method}

The purpose of this section is to provide 
theoretical foundation to the foregoing spectral algorithm.
In particular, we prove that if the angle
between $M$ and $\mathcal{H}(X)$ is ``small'' then
$k$ leading eigenvalues of $M$ are well separated from
the others, their values are close
to the numbers $q(C_1),\ldots,q(C_k)$,  
and the span of the corresponding 
eigenvectors is close to that of $X$.   
To that goal, we extend the technique introduced in
\cite{PerFiedler} to localize a single dominant eigenvalue
of a symmetric matrix.

\subsection{Eigenvalues}

Equation \eqref{eq:maxcos} yields an attainable upper bound
on the cosine between $M$ and any space $\mathcal{H}(X)$.
In the following theorem we prove a refined upper bound 
that provides a stronger indication on the dominance of
the first $k$ eigenvalues of $M$.

\begin{theorem}   \label{thm:sigma_approx}
Let $\{C_1,\dots, C_k\}$ be a family of 
pairwise disjoint subsets of $V$, ordered so that $|q(C_1)|\geq \dots \geq |q(C_k)|$, and let
$X = [\chi_1,\ldots,\chi_k]$.
Let the eigenvalues of $M$ be denoted as in \eqref{eq:lambdaM}.
If $s = \sin(M,\mathcal{H}(X))$ and 
$c = \cos(M,\mathcal{H}(X))$ then
$$
   \sum_{i=1}^k \lambda_i^2 \geq 
   c^2\sum_{i=1}^n \lambda_i^2 + 
   \sum_{i=1}^k (|\lambda_i|-|q(C_i)|)^2 .
$$
Moreover, 
$$
   \sum_{i=1}^k (|\lambda_i|-|q(C_i)|)^2
   \leq s^2\sum_{i=1}^k\lambda_i^2 .
$$
\end{theorem}

\begin{proof}
Note that the eigenvalues of $H = \mathcal P(M,\mathcal{H}(X))$ are the numbers $q(C_i)$,
and the singular values of $M$ and $H$ are the absolute values
of their respective eigenvalues. Thus, using 
Theorem \ref{thm:hoffman-wielandt-sv}, we get 
\begin{equation}\label{eq:a}
    s^2 \sum_{i=1}^n \lambda_i^2 = 
    s^2 \|M\|_\fro^2 = 
    \|M-H\|_\fro^2 \geq \sum_{i=k+1}^n
    \lambda_i^2 + \sum_{i=1}^k (|\lambda_i|-|q(C_i)|)^2 .
\end{equation}
With simple manipulations we get 
$$
   \sum_{i=1}^k \lambda_i^2 \geq
   (1-s^2) \sum_{i=1}^n \lambda_i^2 + 
   \sum_{i=1}^k (|\lambda_i|-|q(C_i)|)^2 ,
$$
proving the first claim. Similarly, from \eqref{eq:a} we obtain 
$$ 
   \sum_{i=1}^k (|\lambda_i|-|q(C_i)|)^2
   \leq s^2 \sum_{i=1}^k \lambda_i^2 - (1-s^2)  
   \sum_{i=k+1}^n\lambda_i^2  
   \leq s^2  \sum_{i=1}^k\lambda_i^2
$$
proving the second claim as well. 
\end{proof}

\begin{remark}
Simple arguments show the equality 
$$
   \frac{(|\lambda_i|-|q(C_i)|)^2}{\lambda_i^2} = 
   \left| 1 - \frac{q(C_i)}{\lambda_i} \right|^2 . 
$$
Hence, if for $i = 1,\dots,k$ we have 
$| 1 - q(C_i)/\lambda_i | > s$ then
$$
   \sum_{i=1}^k (|\lambda_i|-|q(C_i)|)^2 > 
   s^2 \sum_{i=1}^k \lambda_i^2
$$
and the inequality in the last claim of Theorem \ref{thm:sigma_approx} cannot hold. 
Thus, in the hypotheses of that theorem, for at least some
indices $i$ we must have
$1-s \leq q(C_i)/\lambda_i \leq 1+s$.
\end{remark}

\subsection{Eigenspaces}

Our next aim is to evaluate quantitatively the
closeness of $\chi_1,\ldots,\chi_k$ to the 
eigenspace associated to the first
$k$ eigenvalues of $M$. To quantify 
the departure of $\mathrm{Range}(X)$ from
that eigenspace  
we resort to a classical metric among subspaces in $\mathbb{R}^n$
and a particular inequality in the spirit of
Davis--Kahan ``$\sin\theta$'' theorem \cite[Thm.\ 3.6]{MR1061154}.

Let $X,Y\in\mathbb{O}^{n\times k}$ 
and let $\mathcal{X} = \mathrm{Range}(X)$
and $\mathcal{Y} = \mathrm{Range}(Y)$.
Recall that the singular values of $X^TY$ are the 
cosines of the principal angles between $\mathcal{X}$
and $\mathcal{Y}$. These angles quantify the deviation
of each subspace from the other. Moreover,
if $X_\perp$ is an $n\times (n-k)$ matrix such that 
$[X,X_\perp]$ is orthogonal, then the singular values of
$X_\perp^TY$ are the sines of the same angles.
In fact, metrics in the set
of all $k$-dimensional subspaces of $\mathbb{R}^n$ 
are usually defined by means of an unitarily invariant matrix norm
as $\mathrm{dist}(\mathcal{X},\mathcal{Y})
= \|X_\perp^TY\|$, see e.g., \cite[Thm.\ 4.9]{MR1061154}.

\begin{theorem}
Let $M = U\Lambda U^T$ be a spectral decomposition of $M$,
and consider the partitioning $U = [U_1,U_2]$ where
$U_1\in\mathbb{O}^{n\times k}$ is made by unit eigenvectors associated to
the $k$ leading eigenvalues of $M$. 
If $X = [\chi_1,\ldots,\chi_k]$ and $X_\perp$ are as before
then
$$
   \|U_1^TX_\perp\|_\fro^2 \leq
   s^2 \frac{\|M\|_\fro^2}{\lambda_k^2} -
   \sum_{i=1}^k
   \frac{\| M\chi_i - q(C_i)\chi_i\|_2^2}{\lambda_k^2} ,   
$$
where $s=\sin(M,\mathcal{H}(X))$. 
\end{theorem}

\begin{proof}
Let $H = \mathcal{P}(M,\mathcal{H}(X))$ and
$\Lambda_1 = U_1^TMU_1 = \mathrm{Diag}(\lambda_1,\ldots,\lambda_k)$.
By construction $HX_\perp = O$. Consequently,
$$
   U_1^T(M-H)X_\perp = 
   U_1^TMX_\perp = \Lambda_1U_1^TX_\perp .
$$
Moreover,   
$$
   \| (M-H)X_\perp \|_\fro \geq 
   \| U_1^T(M-H)X_\perp \|_\fro  
   = \| \Lambda_1U_1^TX_\perp\|_\fro \geq 
   |\lambda_k| \|U_1^TX_\perp\|_\fro ,
$$
the last inequality coming from
$\|U_1^TX_\perp\|_\fro\leq \|\Lambda_1^{-1}\|_2 \|\Lambda_1U_1^TX_\perp\|_\fro$.
By the orthogonality condition $X^TX_\perp = O$ we also obtain
\begin{align*}
   \| (M-H)X_\perp \|_\fro^2 & = 
   \| M-H\|_\fro^2 - \| (M-H)X\|_\fro^2 \\
   & = \sin^2(M,\mathcal{H}(X)) \|M\|_\fro^2 -
   \sum_{i=1}^k \| M\chi_i - q(C_i) \chi_i\|_2^2 .
\end{align*}
Collecting all results we complete the proof.
\end{proof}

\begin{corollary}   \label{cor:sintheta}
In the notations of the previous theorem,
there exists a $k\times k$ orthogonal matrix $Z$ such that
$$
   \| U_1 - XZ \|_\fro \leq
   \frac{\sqrt{2}\, \|M\|_\fro}{|\lambda_k|} 
   \sin(M,\mathcal{H}(X)) .
$$
\end{corollary}

\begin{proof}
From \cite[Thm. 4.11]{MR1061154} we get
$$
    \min_{Z^TZ = I} \| U_1 - XZ \|_\fro 
    \leq \sqrt{2} \|U_1^TX_\perp\|_\fro ,
$$
and the claim follows from the previous theorem.
\end{proof}

The previous results show that if the angle
between $M$ and $\mathcal{H}(X)$ is small then not only $\mathrm{Range}(X)$ is
close to $\mathrm{Range}(U_1)$ but also the residuals
$\| M\chi_i - q(C_i) \chi_i\|_2$ must be small.

\section{Modularity analysis of the stochastic block model}
\label{sec:SBM}

\newcommand{\ttimes}{\hbox{$\times\hspace{-1mm}\times$}}

A Stochastic Block Model (SBM) with $n$ nodes and $k$ blocks
is a random graph model parametrized by the 
{\em membership matrix} $\Theta\in\{0,1\}^{n\times k}$
and the symmetric {\em connectivity matrix} $B=(b_{ij})\in\mathbb{R}^{k\times k}$. Every row of the matrix $\Theta$ contains exactly one nonzero entry, whose position indicates which block that node belongs to; and the nonzero entries in the $i$-th column 
indicate nodes belonging to the $i$-th block.   
For visual convenience, we assume that each 
block consists of consecutive integers.
The entry $b_{ij}$ is the edge probability
between any node in block $i$ and any node in block $j$.
Edges are generated independently from one another.

The SBM is one of the most widespread generative models   
for random graphs 
and is widely used as a theoretical
benchmark for graph partitioning and community
detection algorithms, see e.g., \cite{ChungCOLT2012,KarrerNewman2011,QinNIPS2013}. 

Suppose that the $i$-th block has $n_i$ elements,
$n_1+\cdots +n_k = n$.
Hence, the average adjacency matrix within the
SBM with parameters $(\Theta,B)$ is 
\begin{equation}   \label{eq:SBM}
   \bar A = \Theta B\Theta^T = \begin{pmatrix}
   \bar A_{11} & \cdots & \bar A_{1k} \\ \vdots & & \vdots \\
   \bar A_{k1} & \cdots & \bar A_{kk}   
   \end{pmatrix} , 
   \qquad
   \bar A_{ij} = b_{ij}\uno\uno^T\in\mathbb{R}^{n_i\times n_j} .
\end{equation}

\begin{lemma}   \label{lem:M_SBM}
The average Newman--Girvan modularity matrix \eqref{eq:NG} 
of the SBM with parameters $(\Theta,B)$ is 
$$
   \bar M = \begin{pmatrix}
   \bar M_{11} & \cdots & \bar M_{1k} \\ \vdots & & \vdots \\
   \bar M_{k1} & \cdots & \bar M_{kk} \end{pmatrix} ,
   \qquad
   \bar M_{ij} = (b_{ij}-\bar d_i\bar d_j/\nu)\uno\uno^T\in\mathbb{R}^{n_i\times n_j} ,
$$
where $\bar d = \bar A\uno$ and $\nu = \uno^T\bar d$.
\end{lemma}

\begin{proof}
By construction, the vector $\bar d = \bar A\uno$ 
with $\bar A$ given by \eqref{eq:SBM} is the average degree vector for 
a random graph extracted from the considered model. 
In particular, if node $i$ belongs to block $\ell$ then
$\bar d_i = \sum_{j=1}^k b_{\ell j}n_j$.
By linearity of the expectation, 
$\bar M = \bar A - \bar d\, \bar d^T\!/\nu$  
for some scalar $\nu$.
Since the equation $M\uno = 0$ holds true for every 
Newman--Girvan modularity matrix, we must set
$\bar M\uno = 0$.
That equation produces the value for $\nu$
indicated above, and the proof is complete.
\end{proof}

We borrow from \cite{FHS} the definition of inflation product of two matrices.

\begin{definition}
Let $A\in\mathbb{R}^{n\times n}$ be a matrix partitioned 
in $k\times k$ block form, with (nontrivial) square diagonal blocks 
having possibly different sizes,   
and let $B=(b_{ij})\in\mathbb{R}^{k\times k}$. The 
{\em inflation matrix} of $B$ with respect to $A$ is
the $n\times n$ block matrix
$$
   B \ttimes A = 
   \begin{pmatrix}
   b_{11} A_{11} & \cdots & b_{1k} A_{1k} \\ \vdots & & \vdots \\
   b_{k1} A_{k1} & \cdots & b_{kk} A_{kk}  
   \end{pmatrix} .
$$
\end{definition} 

The notation $B\ttimes A$ 
does not mention explicitly the partitioning of $A$ on which the
result depends; that partitioning should be clear from the context.
In this section, we always refer to the block partitioning 
appearing in \eqref{eq:SBM}.
Moreover, note that the operator $\ttimes$ defined above is 
a special case of the Khatri--Rao product of 
two block matrices, see e.g., \cite[\S 12.3.3]{GVL},
and is closely related to the Kronecker product $\otimes$; indeed, when $n_1=\ldots = n_k = n/k$ and all blocks of $A$ are equal to $Z$ then $B\ttimes A = B\otimes Z$.
For notational convenience,
we extend the operator $\ttimes$ to vectors as follows:
if $w\in\mathbb{R}^n$ is a vector partitioned into
$k$ (nontrivial) sub-vectors as 
$w = (w_1^T,\ldots,w_k^T)^T$    
and $v= (v_1,\ldots,v_k)^T\in\mathbb{R}^k$ then
$$
   v \ttimes w = 
   (v_1w_1^T,\ldots,v_kw_k^T)^T \in\mathbb{R}^n .   
$$

The following result is a simple case of more general results shown in Lemma 4.17 and Lemma 4.19 of \cite{FHS}; we refrain from including the rather technical proof.
    
\begin{lemma}   \label{lem:eig_xx}
Let $w\in\mathbb{R}^{n}$ be a vector with 
no zero entries,
and let $B\in\mathbb{R}^{k\times k}$ be a symmetric matrix with spectral decomposition $B = V\Lambda V^T$
where $V = [v_1\ldots,v_k]$ and
$\Lambda = \mathrm{Diag}(\lambda_1,\ldots,\lambda_k)$.
Suppose that $w$ is partitioned into $k$ sub-vectors,
$$
   w = (w_1^T,\ldots,w_k^T)^T ,  
   \qquad w_i \in\mathbb{R}^{n_i} ,
$$
such that $w_i^Tw_i = 1$ for $i = 1,\ldots,k$.
Then 
\begin{equation}   \label{eq:xx}
    B\ttimes ww^T = \sum_{i=1}^k
    \lambda_i v_iv_i^T\ttimes ww^T .
\end{equation}
\end{lemma}
Observe that equation \eqref{eq:xx} is actually a spectral decomposition. Indeed, $v_iv_i^T \ttimes ww^T$ is a rank-one matrix that can be written also as $z_iz_i^T$ with 
$z_i = v_i \ttimes w$; furthermore, $z_i^Tz_i = 1$
and $z_i^Tz_j = 0$ for $i\neq j$.

\begin{theorem}   \label{thm:SBM}
Let $\bar M$ be the average Newman--Girvan
modularity matrix of the SBM with parameters $(\Theta,B)$.
Let $N = \mathrm{Diag}(n_1^{1/2},\ldots,n_k^{1/2})$, 
and let $\delta = (\delta_1,\ldots,\delta_k)^T$
where $\delta_i = \sum_{j=1}^k b_{ij}n_j$ for $i = 1,\ldots,k$.
Then the nonzero eigenvalues of $\bar M$ coincide with
the nonzero eigenvalues of the $k\times k$ matrix 
$$
   N (B - \delta\delta^T\!/\nu)N , 
   \qquad \nu = \sum_{i,j}b_{ij}n_in_j .
$$
Furthermore, let $(\lambda_i,v_i)$ be an eigenpair of that matrix
and let $z_i = N^{-1}v_i\ttimes\uno$.
Then $\bar Mz_i = \lambda_iz_i$.
\end{theorem}

\begin{proof}
Observe that the number $\delta_i$ is the average
degree of the nodes belonging to the $i$-th block.
In fact, the average degree vector $\bar d$ introduced
in Lemma \ref{lem:M_SBM} can be written as
$\bar d = \delta \ttimes \uno$. 
By considering the explicit form of $\bar M$ which is given in 
that lemma, it is not difficult to 
recognize that $\bar M = Z \ttimes \uno\uno^T$ where
\begin{equation}   \label{eq:Z}
   Z = B - \delta\delta^T\!/\nu .
\end{equation}
Define $w = (w_1^T,\ldots,w_k^T)^T$    
where $w_i = \uno/n_i^{1/2}\in\mathbb{R}^{n_i}$ and let $Z = (z_{ij})$. 
Then, for $i,j = 1,\ldots,k$ 
the block $(i,j)$ of 
$Z\ttimes \uno\uno^T$ is the $n_i\times n_j$ matrix
$$
   z_{ij}\uno\uno^T = n_i^{1/2}n_j^{1/2} z_{ij}w_iw_j^T
   = (NZN)_{ij}w_iw_j^T = (N ZN \ttimes ww^T)_{ij} ,
$$
hence $\bar M = N ZN \ttimes ww^T$. 
The fist part of the
claim follows by a straightforward application of Lemma \ref{lem:eig_xx}. To complete the proof it is sufficient to observe that the $j$-th sub-vector of $z_i = v_i\ttimes w_i$
is 
$$
   (v_i\ttimes w_i)_j = (v_i)_j \uno /n_j^{1/2}
   = (N^{-1}v_i)_j \uno = (N^{-1}v_i \ttimes \uno)_j 
   \in\mathbb{R}^{n_j} 
$$
for $j = 1,\ldots,k$,
hence $z_i = N^{-1}v_i\ttimes\uno$ and the proof is complete.
\end{proof}

We point out
that the matrix $N(B-\delta\delta^T\!/\nu)N$
has a nontrivial kernel. Indeed, if $v = (n_1,\ldots,n_k)^T$ then
$Bv = \delta$, $\delta^Tv = \nu$ and
$$
   N(B-\delta\delta^T\!/\nu)N\uno = 
   N(B-\delta\delta^T\!/\nu)v = N(\delta - \delta(\delta^Tv/\nu))
   = 0 .
$$
Hence, the matrix $\bar M$ for a SBM with $k$ blocks
has at most $k-1$ nonzero 
eigenvalues.
This fact suggests that the presence of $k$ ``dominant''
modules in a SBM graph is indicated by $k-1$ dominant eigenvalues in the spectrum of the modularity matrix. 
The special case where all those modules have positive modularity
(i.e., they are communities) has been considered in 
Theorem 6.1 of
\cite{fasino2014algebraic}. 

\begin{remark}
Let $\Delta = \mathrm{Diag}(\delta_1,\ldots,\delta_k)$. 
Then, $\bar W = \Delta \ttimes I$
is the diagonal matrix of the average degrees in the SBM.
By means of arguments completely analogous to the preceding ones,
we can also consider the
matrix $\widetilde M = \bar W^{-1/2}\,\bar M\,\bar W^{-1/2}$ that,
within some approximation, can be considered as the
average modularity matrix associated to the second 
case of Example \ref{ex:NG}.
The results below are straightforward:
\begin{itemize}
\item
$\widetilde M = \Delta^{-1/2}Z\Delta^{-1/2}\ttimes \uno\uno^T$ 
where $Z$ is as in \eqref{eq:Z};

\item
equivalently, $\widetilde M = N\Delta^{-1/2}Z\Delta^{-1/2}N\ttimes ww^T$ 
where $w$ is as in the proof of Theorem \ref{thm:SBM}; 

\item
let $(\lambda_i,v_i)$ be an eigenpair of $N\Delta^{-1/2}Z\Delta^{-1/2}N$
and let $z_i = N^{-1}v_i\ttimes\uno$.
Then $\widetilde Mz_i = \lambda_iz_i$,
and $\widetilde M$ has no other nonzero eigenvalues.   
\end {itemize}
Note that the eigenvectors $z_1,\ldots,z_k$ found above
of both $\bar M$ and $\widetilde M$
are constant within each block of the SBM. 
Consequently, if $\mathcal{C} = \{C_1,\ldots,C_k\}$
covers $V$ and $X = [\chi_1,\ldots,\chi_k]$ then
$\cos(\bar M,\mathcal{H}(X)) = 
\cos(\widetilde M,\mathcal{H}(X)) = 1$.
\end{remark}

\section{Numerical examples}   \label{sec:numerical}

In this section we apply the spectral method discussed so far 
to the simultaneous search  
for communities and anti-communities in some 
real world
networks of moderate size. The main aim of the following experimental analysis is to show the 
effectiveness of the method in recognizing different  
components of a   
network being either well inter-connected or well intra-connected, by exploiting the invariant subspace associated to the eigenvalues with largest modulus  
of the modularity matrix. 
We do not concern ourselves with 
implementation details and performance issues.

In the forthcoming experiments 
we consider graphs with unweighted edges, $w\equiv 1$, 
and we fix the vertex weight function 
$\mu\equiv 1$. This leads to the original modularity formulation 
\eqref{eq:NG}.
We locate 
a small number $k$ of dominant eigenvalues of the modularity 
matrix that are well separated from the others
on the basis of a visual inspection of the spectrum.
More precisely, considering the eigenvalues 
numbered as in \eqref{eq:lambdaM},
we plot the ratios $|\lambda_i|/|\lambda_{i+1}|$
for $i = 1,\ldots,20$; the index $k$  
is chosen by discriminating those ratios that 
stand out from the others, see Figure \ref{fig:eigs}.
Then, we collect
the associated eigenvectors in an $n\times k$
matrix; the rows of that matrix are 
partitioned into $k+1$ groups
by means of the Matlab function {\tt kmeans},
and the resulting clusters are transferred to the original network.

\begin{figure}[ht]
\begin{center}
\includegraphics[width=0.45\linewidth]{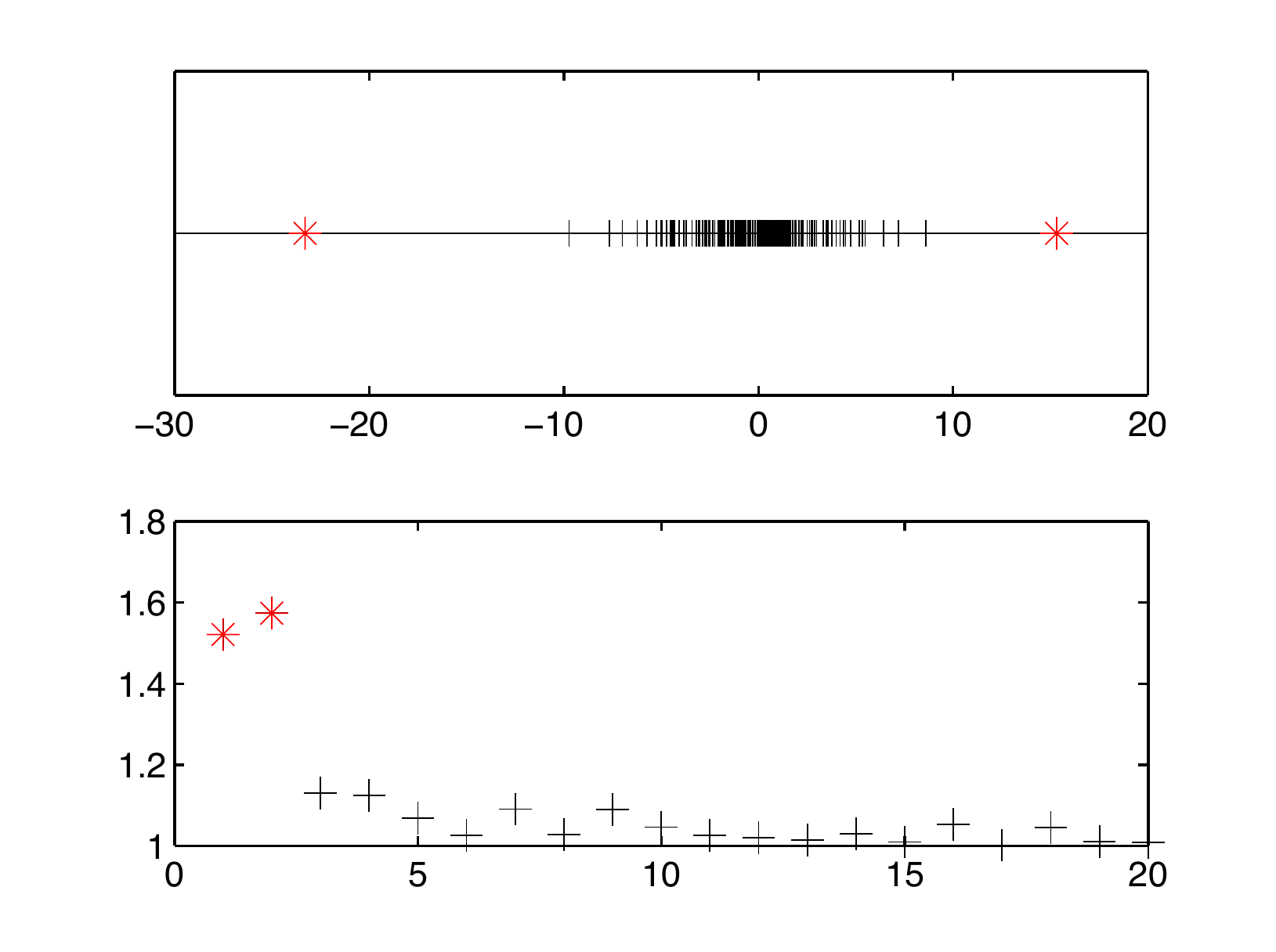}
\includegraphics[width=0.45\linewidth]{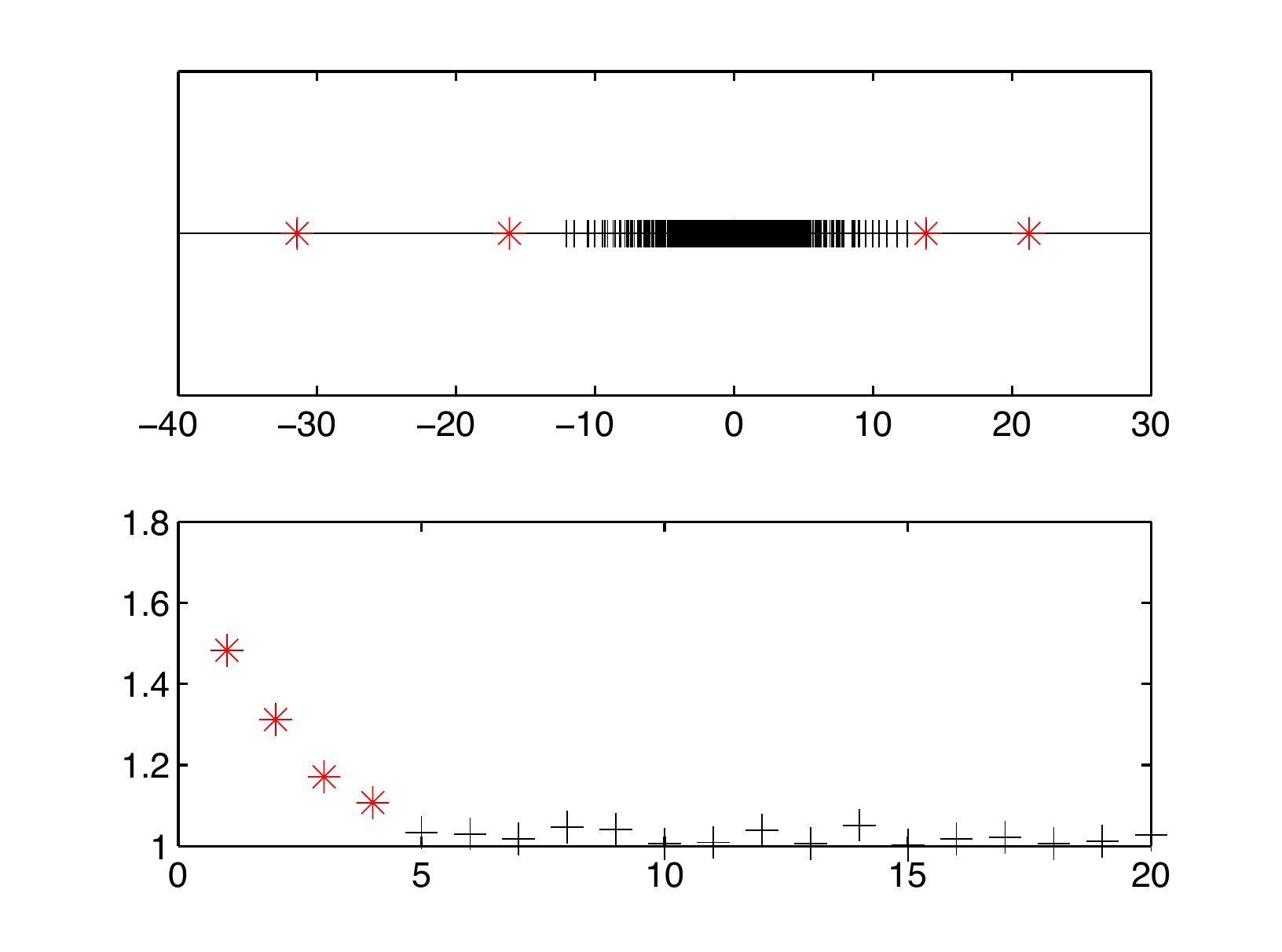}
\end{center}
\caption{
Modularity 
eigenvalue distributions (top) 
and ratios $|\lambda_i/\lambda_{i+1}|$ (bottom)
of two real world networks. 
Left: \textit{Small World} citation network. Right: {\em E.\ Coli} protein network. Eigenvalues marked with a red asterisk are considered as relevant.}
\label{fig:eigs}
\end{figure}

\subsubsection*{Small World citation network} 

The first example we discuss 
is an instance of the network known as ``Small world citation network'' where the nodes represent   
papers that cite the pioneering  work of  S.\ Milgram \cite{travers1967small} or contain ``Small World'' in the title, appeared in the period 1967--2003,
and edges represent citations. 
This network contains 233 nodes and 994
oriented edges and is drawn from the Garfield's collection of citation network datasets produced using the 
HistCite software \cite{garfield2004histcite}. The adjacency matrix, which is freely available in Matlab format from \cite{davis2011university}, has been symmetrized
by neglecting edge orientation.
Looking at the two leftmost
plots in Figure \ref{fig:eigs} one clearly recognizes two outliers in the spectrum of the modularity matrix. 
Thus we embed the nodes into the plane 
spanned by the eigenvectors  
associated with the two eigenvalues with largest magnitude and apply $k$-means to partition the node set into three groups, $\mathcal C = \{C_1,C_2, C_3\}$. The method clearly identifies two communities and one anti-community, as shown in Table \ref{tab:smallw} and Figure \ref{fig:smallw}. 
\begin{table}[ht]
\centering
 \begin{tabular}{c|ccc}
  & $C_1$ & $C_2$ & $C_3$\\
  \hline
$|C|$ & 104 & 2 & 127 \\    
$Q(C)$ & 336 & $-$30 & 339 \\
$q(C)$ & 3.23 & $-$15 & 2.67\\
\end{tabular} 
\caption{Size, modularity and normalized modularity values for the sets $C_1$, $C_2$ and $C_3$, located by the spectral method on the  
\textit{Small World} citation network.}\label{tab:smallw}
\end{table}

\begin{figure}[ht]
\begin{center}
\includegraphics[width=0.3\linewidth, trim =  31mm 9mm 25mm 7mm, clip]{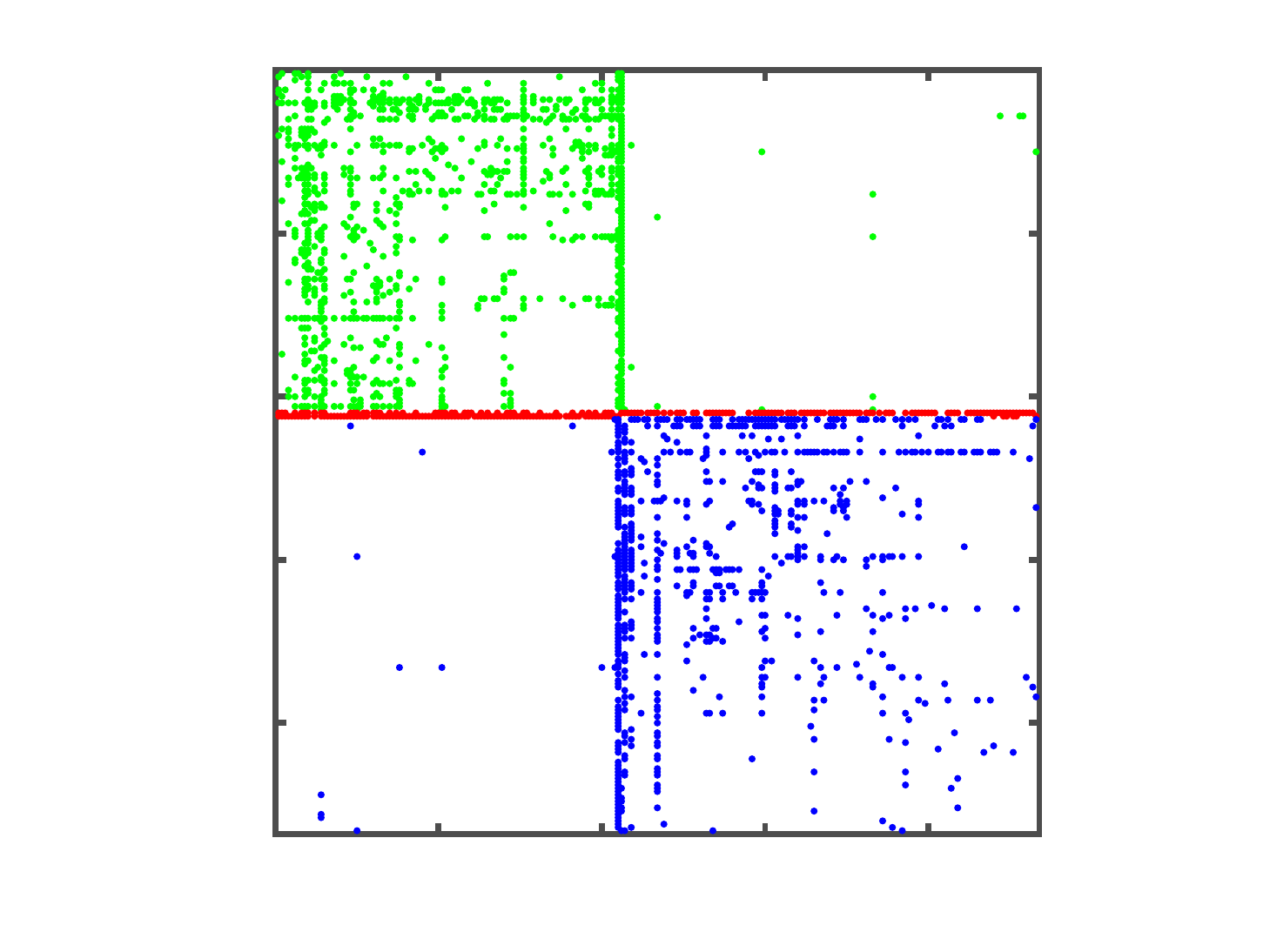}
\qquad
\includegraphics[width=0.3\linewidth, trim =  12mm 10mm 18mm 8mm, clip, angle=0]{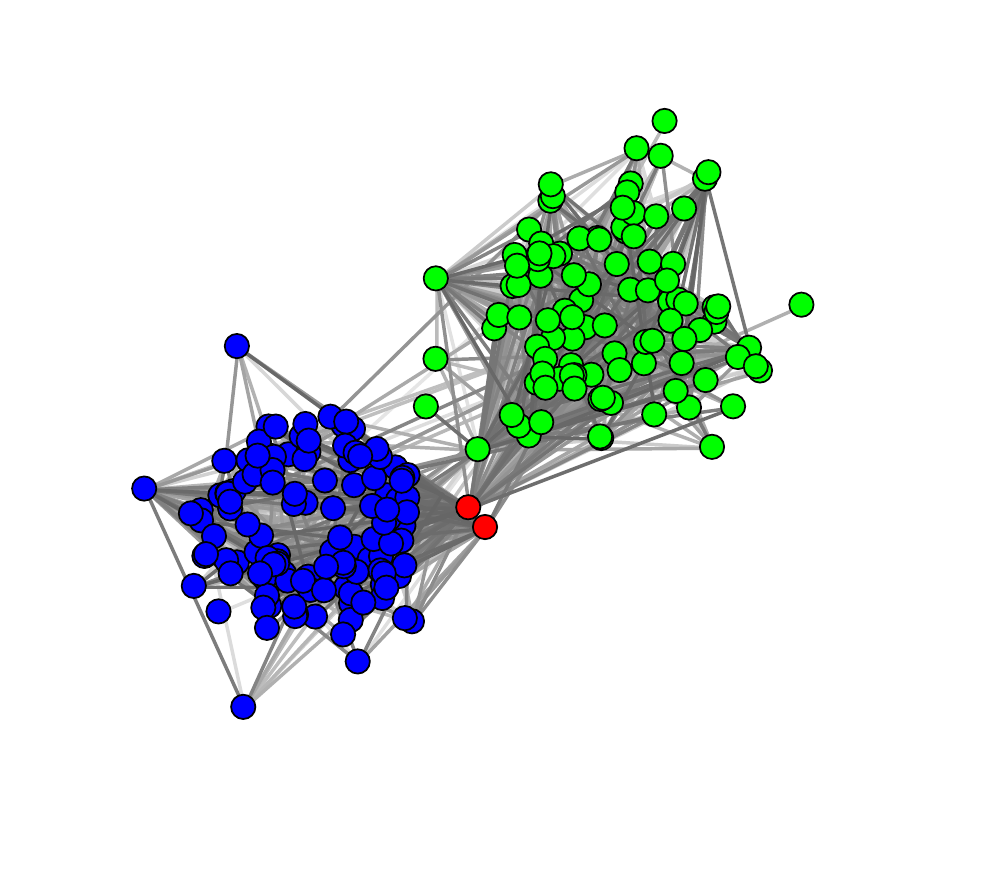}
\end{center}
\caption{(Best seen in colors) Simultaneous community/anti-community 
detection on the  
\textit{Small World} citation network. 
Left: colored sparsity pattern \textit{(spy)} of the adjacency matrix, when the  
nodes of the network are clustered using the first two eigenvectors of $M$. 
Right: graph drawing of the communities $C_1$ and $C_3$ (blue and green) and the anti-community $C_2$ (red) located by the method.}\label{fig:smallw}
\end{figure}

\subsubsection*{E.\ Coli protein-protein interaction network} 

The second 
network shows the interaction between  proteins  within the  \textit{Escherichia Coli} bacterium. Several instances of this network have been gathered in recent years. Here we consider the dataset from \cite{wuchty2014protein} formed by  10788 edges (interactions) between 1251 nodes (proteins).
The relative separation of the first $4$ 
modularity eigenvalues 
is rather different from that of the bulk of the spectrum,
as shown in the rightmost plots in Figure \ref{fig:eigs},
so we set $k = 4$ in this example.
Using $k$-means to locate a partition into 5 groups,
we are able to capture a clear community and anti-community structure in this network. The method locates the partition $\mathcal C=\{C_1,\dots,C_5\}$ shown by the colored sparsity patterns in Figure \ref{fig:pinecoli}. A large portion of the nodes are assigned to a single large community $C_1$ (colored  in blue) which however has a moderate modularity value $q(C_1)$ and can be regarded as ``background noise'',  
whereas the remaining nodes are well-clustered into two  communities $C_2$ and $C_3$ (colored in red and green, respectively) and two anti-communities $C_4$ and $C_5$ (colored in magenta and cyan, respectively).  
Table \ref{tab:pinecoli} shows sizes and modularity scores for the node sets in $\mathcal C$, whereas the 
rightmost plot in Figure \ref{fig:pinecoli} 
provides a qualitative overview of the community/anti-community structure of the 
most relevant part of the graph.

\begin{table}[ht]
\centering
\begin{tabular}{c|ccccc}
  & $C_1$ & $C_2$ & $C_3$ & $C_4$ & $C_5$\\
  \hline
$|C|$ & 1021 & 86 & 60 & 50 & 34 \\    
$Q(C)$ & 377 & 396 & 289 & $-$75 & $-$335 \\
$q(C)$ & 0.36 & 4.61 & 4.82 & $-$1.51 & $-$9.87\\
\end{tabular} 
\caption{Size, modularity and normalized modularity values for the sets $C_1, \dots, C_5$, located by the spectral method on the  
\textit{E. Coli} protein-protein interaction network.}\label{tab:pinecoli}
\end{table}

\begin{figure}[ht]
\begin{center}
\includegraphics[width=0.3\linewidth, trim =  34mm 9mm 25mm 12mm, clip]{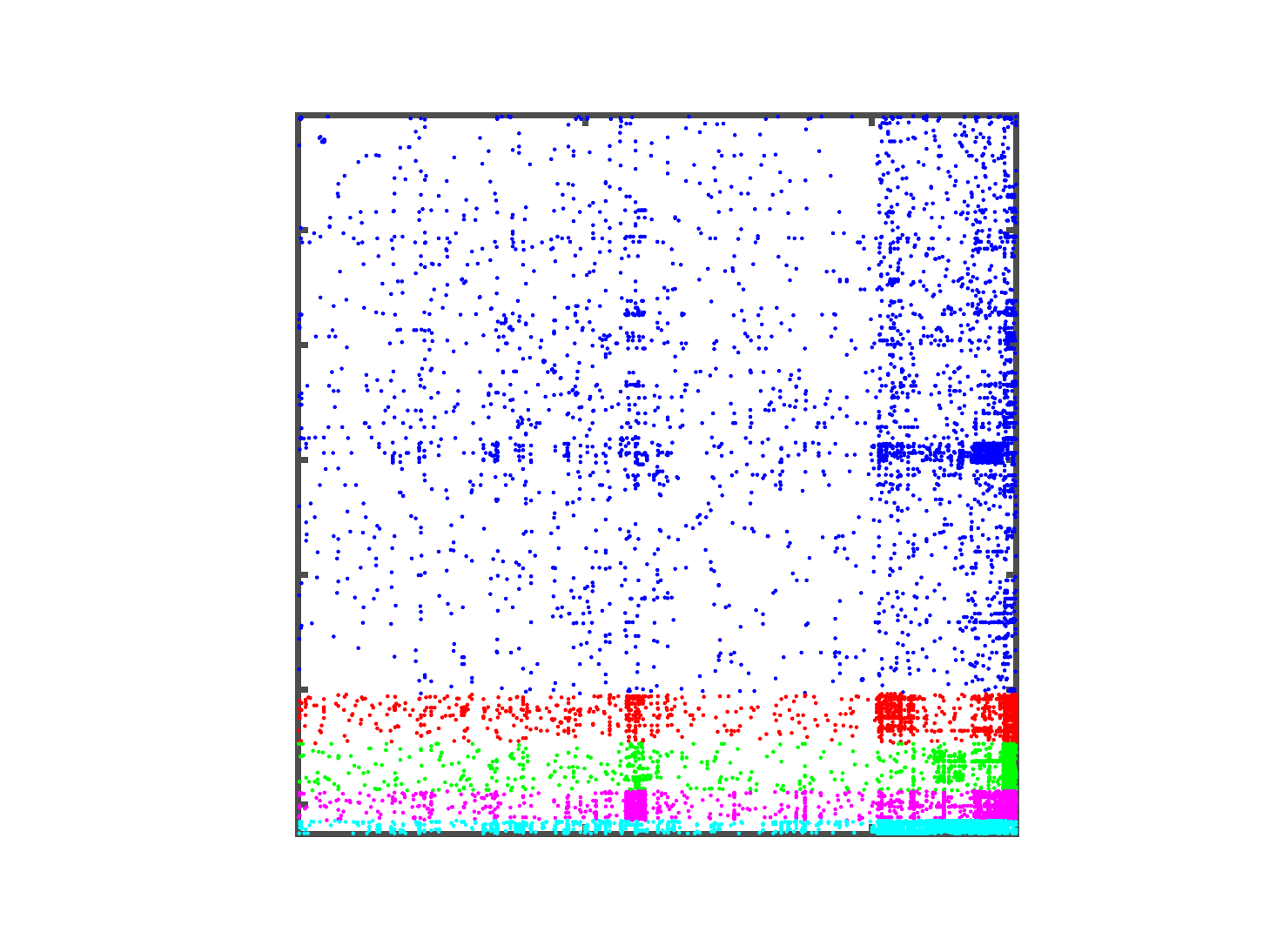}
\quad
\includegraphics[width=0.3\linewidth, trim =  34mm 9mm 25mm 12mm, clip]{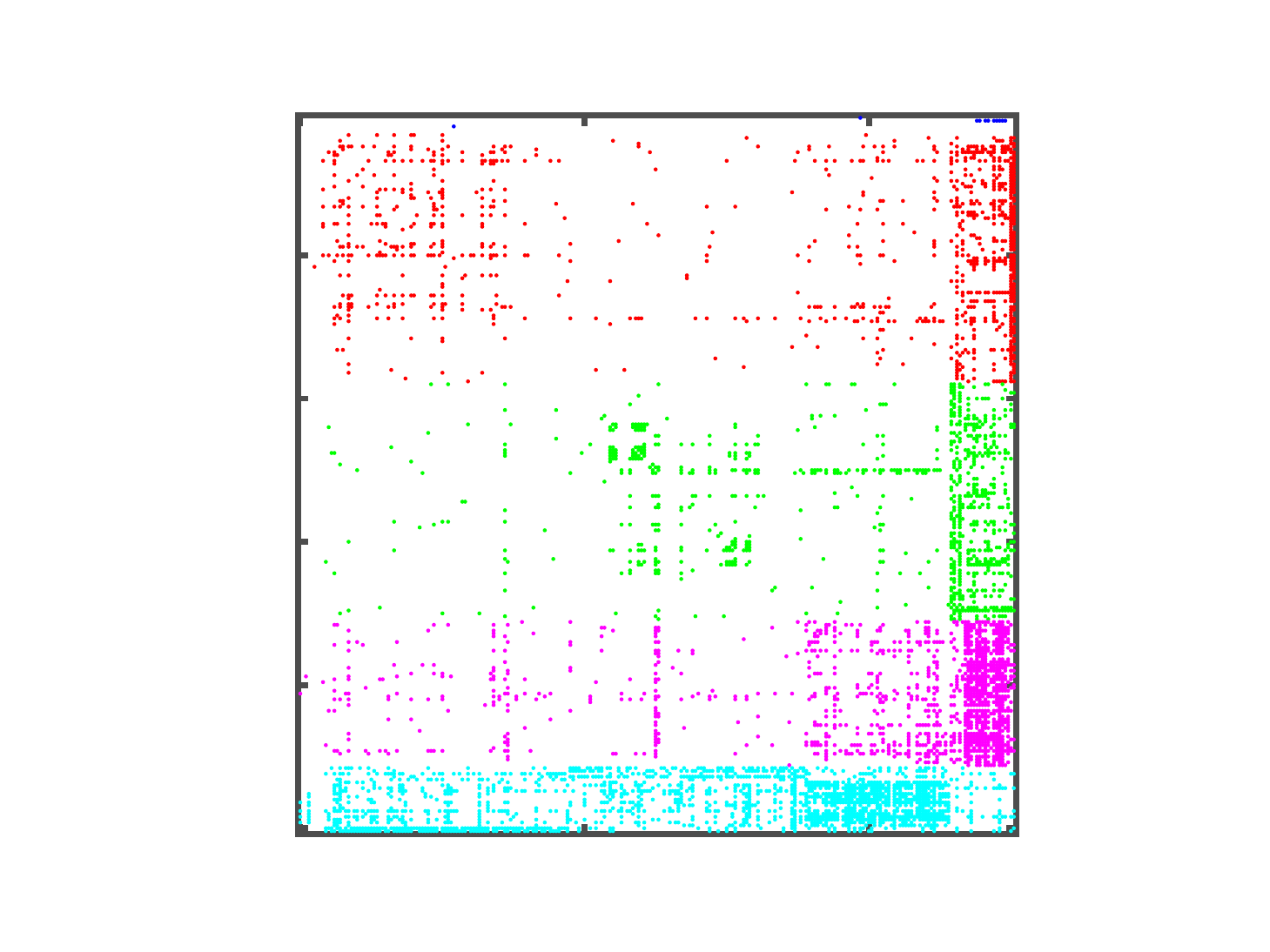}
\quad
\includegraphics[width=0.3\linewidth, trim =  54mm 33mm 41mm 29mm, clip]{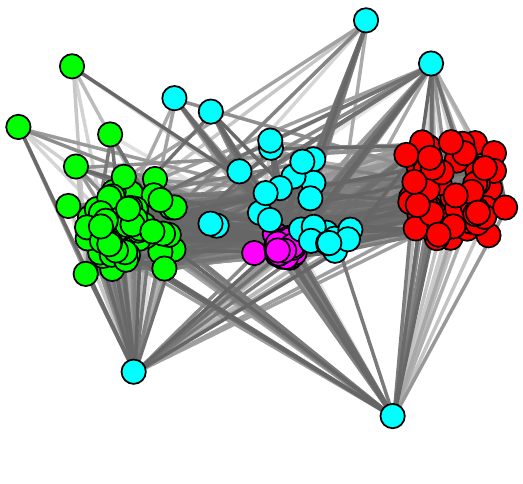}
\end{center}
\caption{(Best seen in colors) Simultaneous community/anti-community detection on 
the  
\textit{E.\ Coli} protein-protein interaction network. 
Left: colored sparsity pattern \textit{(spy)} of the adjacency matrix, when the  
nodes of the network are clustered using four dominant eigenvectors of $M$. Center: close-up of the  
entries of the adjacency matrix corresponding to the four smaller clusters $C_2,\dots,C_5$.
Right: graph drawing of the communities $C_2$ and $C_3$ (red and green) and the anti-communities $C_4$ and $C_5$ (magenta and cyan) located by the method.} \label{fig:pinecoli}
\end{figure}

\section{Conclusions}

Anti-communities are group of nodes
showing few internal connections but being highly connected with the rest of the graph. Motivated by a number of data science applications, the interest towards  this kind of 
groups is growing alongside the one for more classical communities. 

In this work we propose the use of extremal eigenvalues and eigenvectors of generalized modularity matrices to simultaneously look for non-overlapping group of nodes that are likely to be recognizable as communities or anti-communities in a network. 
That technique arises as a  
continuous relaxation of 
the combinatorial optimization problem of maximizing 
the sum of squared normalized modularities.

Our approach is not bound to one specific definition of modularity,
and allows for different normalization terms.
We provide  
other matrix theoretical evidences of the soundness of the spectral method proposed, together with a detailed analysis of the stochastic block model (SBM). Even though spectral approaches based on Laplacian or adjacency matrices  
have been proposed in the past
for this generative model, 
the use of generalized modularity matrices in this context has been mostly overlooked. Our analysis together with our final numerical examples, instead, show the  
effectiveness of this technique
on real world data. 
Furthermore, the machinery we developed for our results can be employed to extend our analysis to the more  
flexible  
{\em degree corrected SBM}
\cite{KarrerNewman2011,QinNIPS2013}
which allows the generation of
random graphs having greater degree heterogeneity
and may yield a better understanding of the behavior of 
spectral methods on real world data.

\end{document}